\documentclass[conference]{IEEEtran}
\IEEEoverridecommandlockouts

\usepackage{cite}
\usepackage{amsmath,amssymb,amsfonts}
\usepackage{amsthm}
\usepackage{algorithmic}
\usepackage{graphicx}
\usepackage{textcomp}
\usepackage{xcolor}
\usepackage[linesnumbered,ruled,vlined,noend]{algorithm2e}
\usepackage{enumitem}
\usepackage{multirow}
\usepackage{subcaption}
\usepackage{xspace}
\usepackage{booktabs}
\usepackage{hyperref}

\newcommand{\spara}[1]{\smallskip\noindent{\bf #1}}
\newcommand{\RAM}{\textsf{RAM}}
\newcommand{\BGD}{\textsf{BGD}}

\newcommand{\MGA}{\textsf{MGA}}

\newcommand{\single}{\textsc{SingleLevelGreedy}\xspace}

\newtheorem{theorem}{Theorem}
\newtheorem{definition}{Definition}
\newtheorem{problemDefinition}{Problem Definition}
\newtheorem{property}{Property}
\newtheorem{proposition}{Proposition}
\newtheorem{example}{Example}
\begin{document}

\title{Representation-Aware Modularity with a Distance-Calibrated Null Model}

%

\author{
\IEEEauthorblockN{Taejoon Han}
\IEEEauthorblockA{
\textit{Ulsan National Institute of}\\
\textit{Science and Technology}\\
Ulsan, Republic of Korea\\
taejoon.han@unist.ac.kr
}
\and
\IEEEauthorblockN{Dahee Kim}
\IEEEauthorblockA{
\textit{Ulsan National Institute of}\\
\textit{Science and Technology}\\
Ulsan, Republic of Korea\\
dahee@unist.ac.kr
}
\and
\IEEEauthorblockN{Junghoon Kim\textsuperscript{*}}
\IEEEauthorblockA{
\textit{Ulsan National Institute of}\\
\textit{Science and Technology}\\
Ulsan, Republic of Korea\\
junghoon.kim@unist.ac.kr
}
\and
\IEEEauthorblockN{Taeho Ha}
\IEEEauthorblockA{
\textit{Soonchunhyang University}\\
Asan, Republic of Korea\\
hth1690@sch.ac.kr
}
}





%

\maketitle

\begingroup
\renewcommand{\thefootnote}{*}
\footnotetext{Corresponding author}
\endgroup

\begin{abstract}
Modularity is a standard objective for community detection, comparing observed intra-community connectivity with its null-model expectation. Modern node embeddings capture structural and semantic proximity beyond raw topology, but replacing adjacency with representation-induced affinity while retaining the classical degree-based null model makes the observed term geometry-dependent and the expected term geometry-agnostic. This creates a null-model mismatch: the objective may reward affinity already expected from representation proximity. This paper introduces {\RAM}, a representation-aware modularity objective with a distance-calibrated null model. {\RAM} preserves the observed-minus-expected principle by calibrating expected affinity with both degree heterogeneity and representation distance. We analyse its relation to classical modularity and develop divisive and multi-level greedy optimisers. Experiments on seven real-world datasets and synthetic scalability tests show that {\RAM} improves over topology-only, direct embedding-clustering, and deep learning-based baselines, remains competitive under classical modularity while improving community quality on large graphs, and scales efficiently. Ablations confirm that both degree correction and distance calibration are essential.
\end{abstract}







\section{INTRODUCTION}

Community detection is a central problem in graph mining, data mining, and network science~\cite{jin2021survey}. Given a graph, the goal is to identify groups of nodes, called communities, such that nodes within the same group are more densely connected to each other than to nodes in other groups~\cite{fortunato2010community}. Among the many formulations proposed for this task, modularity~\cite{newman2004finding,newman2006finding} remains one of the most widely adopted. Classical modularity~\cite{newman2004finding} measures the difference between the observed fraction of intra-community edges and the fraction expected under a null model. Its statistical interpretability has made modularity a standard and well-studied objective function for unsupervised community detection~\cite{fortunato2010community}.

Recent advances in graph representation learning have made node representations a common input to graph mining tasks~\cite{hamilton2020graph}. Such representations can be obtained from graph neural networks~\cite{kipf2016semi}, random-walk embeddings~\cite{perozzi2014deepwalk, grover2016node2vec}, spectral methods~\cite{belkin2003laplacian}, or feature mappings derived from external data~\cite{hamilton2017inductive}. They capture information beyond raw adjacency, such as latent structural similarity, semantic proximity, and higher-order relational patterns~\cite{su2022comprehensive}. However, using representations in modularity optimisation is not merely a matter of replacing edges with representation-induced affinities. Classical modularity is meaningful because its observed term and expected term are defined over the same structural signal. Once the observed affinity is derived from node representations, it becomes dependent on the geometry of the representation space. The degree-based null model of classical modularity does not account for this geometry. This raises a key question: how can modularity preserve its observed-minus-expected semantics when the observed affinity is representation-dependent?

A naive representation-aware extension of modularity replaces the adjacency term with representation-induced affinity while retaining the classical degree-based expectation. This formulation is not fully aligned with the observed-minus-expected interpretation of modularity. If two nodes are close in the representation space, high affinity may be expected even without evidence of community structure beyond their geometric proximity. Conversely, an observed edge between distant representations may carry a different meaning from an edge between nearby representations. A degree-only null model does not explicitly separate these effects. Consequently, the resulting objective may reward affinity that is already expected from representation geometry, rather than affinity that exceeds a proper null-model baseline. We refer to this issue as a null-model mismatch in representation-aware modularity. The mismatch arises because the observed term is geometry-dependent, whereas the expected term remains geometry-agnostic. Therefore, representation-aware modularity motivates recalibrating the null model, rather than only redefining the observed affinity.

\begin{figure*}
\vspace{-0.5cm}
    \centering
    \includegraphics[width=0.99\textwidth]{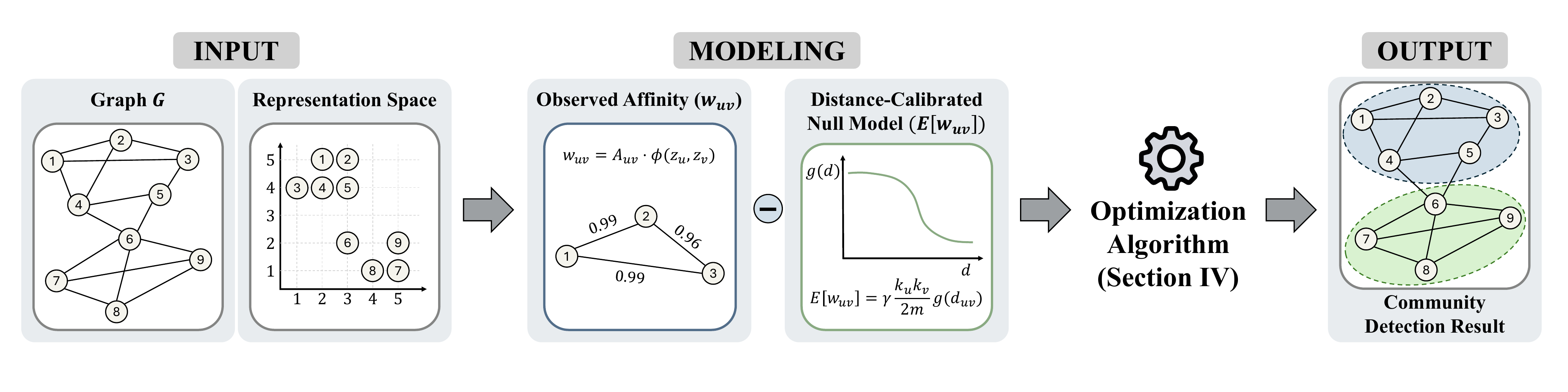}
    \vspace{-0.7cm}
    \caption{Overview of the proposed representation-aware modularity framework}
    \label{fig:overview}
    \vspace{-0.4cm}
\end{figure*}

To address this mismatch, we propose \textit{representation-aware modularity} ({\RAM}), a modularity objective with a distance-calibrated null model. The key idea is to preserve the observed-minus-expected principle of classical modularity under representation-dependent affinity. {\RAM} defines observed affinity on graph-supported node pairs and calibrates the expected term by both degree heterogeneity and representation distance. More specifically, the null model explicitly adjusts the classical degree-based expectation using a distance-dependent calibration function $g(d_{uv})$, where $d_{uv}$ denotes the distance between the representations of nodes $u$ and $v$. The function $g(\cdot)$ acts as the expected similarity conditioned on representation distance, not an arbitrary correction factor. As a result, {\RAM} rewards within-community affinity only when it exceeds what is expected from both node degrees and representation distance.

This formulation positions {\RAM} as an extension of classical modularity, rather than as an embedding-based clustering method. Under the special case where representation-induced affinity reduces to graph adjacency and $g(\cdot)$ is constant across all distances, {\RAM} reduces to classical modularity. Thus, {\RAM} remains compatible with the classical objective while extending its null-model principle to representation-aware community detection.

Figure~\ref{fig:overview} summarises the {\RAM} pipeline: representation-induced affinity defines the observed term, while the distance-calibrated null model defines the corresponding representation-aware expected term.

Optimising {\RAM} is non-trivial. The objective inherits the combinatorial difficulty of modularity maximisation, and the calibrated expected term introduces additional computational challenges. To address this, we develop two strategies. First, \textbf{Betweenness-Guided Divisive Search Algorithm} ({\BGD}) is a top-down approach that iteratively partitions the graph by removing high betweenness edges and evaluates each resulting partition using {\RAM}. Second, \textbf{Multi-Level Greedy Aggregation Algorithm} ({\MGA}) is a hierarchical bottom-up approach that alternates between node-level local moves and contracting communities into super-nodes.

The main contributions of this paper are as follows:
\begin{itemize}[leftmargin=*, topsep=0pt]
    \item We identify a null-model mismatch in representation-aware modularity and formulate {\RAM}, a distance-calibrated objective with a representation-aware expected term.
    \item We preserve the observed-minus-expected principle by calibrating expected affinity with both degree heterogeneity and representation distance.
    \item We instantiate {\RAM} with two optimisation strategies: a divisive search method and a Louvain-style multi-level optimiser with {\RAM}-specific gain updates and centroid summaries.
\end{itemize}
\section{REPRESENTATION-AWARE MODULARITY}
In this section, we formally introduce representation-aware modularity ({\RAM}), a modularity formulation that extends classical modularity by incorporating node representations while preserving the observed-minus-expected principle.

\subsection{Preliminaries}
Let $G = (V, E)$ be an undirected graph, where $V$ denotes the set of nodes with $n = |V|$ and $E$ denotes the set of edges with $m = |E|$. For a node $u\in V$, let $N(u)=\{v\in V \mid (u,v)\in E\}$ denote the set of its neighbours and $k_u=|N(u)|$ represent its degree.

\spara{Classical modularity.}
Community detection aims to partition $V$ into a set of communities $\mathcal{C} = \{C_1, \dots, C_k\}$. Classical modularity~\cite{newman2004finding} evaluates a partition by comparing the observed intra-community connectivity against the expected connectivity under the null model:
\begin{align*}
Q(\mathcal{C}) = \frac{1}{2m} \sum_{C \in \mathcal{C}} \sum_{u,v \in C} \left( A_{uv} - \gamma\frac{k_u k_v}{2m} \right)
\end{align*}
, where $A_{uv}$ is the $(u,v)$ entry of the adjacency matrix $A$, $\frac{k_u k_v}{2m}$ 
is the expected number of edges between $u$ and $v$ under the null model~\cite{newman2004finding}, and $\gamma > 0$ is a resolution parameter controlling the granularity of the resulting  partition~\cite{reichardt2006statistical}. A partition is considered high-quality when the observed connectivity exceeds the expectation, yielding a large positive surplus.

\spara{Node representation.}
Each node $v \in V$ is associated with a $d$-dimensional vector $z_v \in \mathbb{R}^d$ in a representation space, collectively denoted by $\mathcal{Z}=\{z_v\}_{v\in V}$.
These representations can be obtained from various sources, such as graph embedding models~\cite{perozzi2014deepwalk,grover2016node2vec}, spectral methods~\cite{belkin2003laplacian}, or feature mappings derived from external data~\cite{hamilton2017inductive}.
Our framework treats the representations as input and is compatible with different representation learning pipelines.

Our primary objective is to extend the classical modularity framework to graphs equipped with node representations. To this end, we redefine the concept of intra-community surplus to account for both observed structural connectivity and geometric proximity in the representation space, while preserving the observed-minus-expected principle through a distance-calibrated null model.

\subsection{Edge Affinity Based on Representations}

\spara{Motivation.}
As discussed above, evaluating intra-community surplus requires considering both structural connectivity and geometric proximity. Relying solely on edge existence ignores the geometric proximity between nodes in the representation space. Therefore, we introduce a representation-induced edge affinity that assigns a weight to each observed edge based on the geometric proximity. This enables a more informative measure of intra-community connectivity that reflects both topology and representation.

\spara{Construction.}
To incorporate representation information while preserving the edge structure of the graph, we define a representation-induced affinity only on observed edges $(u,v) \in E$:
\begin{align*}
    w_{uv} = A_{uv} \cdot \phi(z_u, z_v)
\end{align*}
, where $A_{uv}$ is the $(u,v)$ entry of the adjacency matrix $A$, and $\phi(\cdot,\cdot)$ is a non-negative similarity function on the representation space.

We restrict affinity to observed edges by setting $w_{uv} = 0$ for $(u,v) \notin E$. This prevents node pairs from appearing connected solely based on their proximity in the representation space, keeping the observed affinity graph-supported.
Let $2m_w = \sum_{u, v \in V} w_{uv}$ denote the total affinity mass of the graph. Unlike the structural edge mass $2m=\sum_{u,v}A_{uv}$, which depends purely on topology, $2m_w$ incorporates representation-induced affinities on observed edges.

\subsection{Distance-Calibrated Null Model}

\textbf{Motivation.}
Classical modularity uses a degree-based null model, which is appropriate when community quality is evaluated purely from graph topology.
In our framework, however, the observed affinity $w_{uv}=A_{uv}\phi(z_u,z_v)$ also depends on representation geometry, so a degree-only null model cannot capture geometry-induced variation in affinity.

\spara{Construction.}
To evaluate the expected affinity, the null model must account for structural and geometric factors. We factorise it into two components: a degree-based term inherited from classical modularity, and a distance-based term that captures the expected similarity induced by representation geometry. This product form should be understood as a degree-corrected calibration baseline, rather than an independence assumption between degree and representation geometry. The degree term $\frac{k_u k_v}{2m}$ determines the interaction opportunity assigned by the null model, while the calibration function $g(d_{uv})$, with $d_{uv}=\|z_u-z_v\|$, determines the expected affinity level at the corresponding representation distance.
Combining the two, the expected affinity is defined as: $\mathbb{E}[w_{uv}] \;=\; \gamma\,\frac{k_u k_v}{2m}\,g(d_{uv})$\label{eq:null-term}, where $\gamma > 0$ is a standard resolution parameter.

Formally, $g(d_{uv})$ is the expected similarity at representation distance $d_{uv}$ conditioned on the chosen similarity function $\phi$:
\begin{align*}
    g(d_{uv}) \;=\; \mathbb{E}[\phi(z_u, z_v) \big|\; \|z_u - z_v\| = d_{uv},\; u \neq v]
\end{align*}
The role of $g(\cdot)$ depends on the choice of similarity function $\phi$.
When $\phi$ captures directional information (e.g., cosine similarity), two node pairs at the same distance may have different similarity values depending on the directions of their embedding vectors.
In this case, $g(\cdot)$ becomes a genuine conditional expectation: it averages out the directional variation and captures how similarity typically changes with distance.
 
\spara{Interpretation of the null term.} Unlike classical modularity, the calibrated null mass in {\RAM} is not required to sum to $2m_w$. The degree-based null model conserves total expected edge mass because node degrees are fixed and $\sum_u k_u=2m$. By contrast, representation proximity is not a conserved quantity: many node pairs can simultaneously be close in the representation space. Thus, $\sum_{u,v}\mathbb{E}[w_{uv}]$ need not match the observed affinity mass. This reflects the role of $g(\cdot)$ as a distance-conditioned expected similarity, not a mass-preserving baseline.

\subsection{Estimation of $g(d)$.}
Computing $g(\cdot)$ exactly requires knowledge of the conditional distribution of $\phi$, which is generally intractable in practice. We therefore estimate $g(\cdot)$ empirically from sampled node pairs.

We sample a set $\mathcal{S}$ of node pairs $(u,v)$ uniformly at random with $u \neq v$, and compute their distances $d_{uv} = \|z_u - z_v\|$ and similarities $\phi(z_u, z_v)$. We partition the observed distance range into $L$ bins $\{\mathcal{B}_1, \dots, \mathcal{B}_L\}$ with centres $\{c_1, \dots, c_L\}$, and for each bin compute the average similarity:
\begin{align*}
    \hat{g}_\ell
    = \frac{1}{|\mathcal{B}_\ell|}
      \sum_{(u,v) \in \mathcal{B}_\ell} \phi(z_u, z_v)
\end{align*}
Each $\hat{g}_\ell$ serves as the representative similarity at the bin centre $c_\ell$. For any query distance $d$ that falls between two adjacent bin centres $c_\ell$ and $c_{\ell+1}$, we estimate $g(d)$ by linearly interpolating $\hat{g}_\ell$ and $\hat{g}_{\ell+1}$. This yields a piecewise linear estimator $\hat{g}(d)$, which is $L_{\hat g}$-Lipschitz, where $L_{\hat g}$ is its maximum slope.

\subsection{Representation-Aware Modularity (\RAM)}\label{sec:ram}
By combining the representation-induced edge affinity with the distance-calibrated null model, we define a new modularity formulation for graphs equipped with node representations.

\begin{definition}[Representation-Aware Modularity (\RAM)]
\label{def:ram}
Given a graph $G=(V,E)$ with node representations $\mathcal{Z}$ and a partition $\mathcal{C} = \{C_1, \dots, C_k\}$ of $V$, the representation-aware modularity is defined as
\begin{align*}
Q_{\mathrm{RA}}(\mathcal{C}) = \frac{1}{2m_w} \sum_{C \in \mathcal{C}} \sum_{u,v \in C} \left( w_{uv} - \gamma \frac{k_u k_v}{2m} g(d_{uv}) \right)
\end{align*}
\end{definition}

The inner summation ranges over all ordered node pairs within each community. Since the graph is undirected, each internal edge is counted twice, which correctly aligns with the global scaling factor $\frac{1}{2m_w}$. The product form in the expected term decomposes the null model into two factors: $k_u k_v/(2m)$ is the expected number of edges between $u$ and $v$ under the degree-based null model of classical modularity, while $g(d_{uv})$ captures the expected similarity at representation distance $d_{uv}$. This factorisation does not assume that topology and representation geometry are independent in the observed graph; rather, it defines a baseline determined by degree heterogeneity and distance-conditioned similarity, against which the observed graph-supported affinity is evaluated.

For notational convenience, let $p_{uv} = \gamma \frac{k_u k_v}{2m} g(d_{uv})$ denote the distance-calibrated null term. We define the observed and expected affinity sums of a community $C$ as $w(C) = \sum_{u,v \in C} w_{uv}$ and $p(C) = \sum_{u,v \in C} p_{uv}$, respectively. $Q_{\mathrm{RA}}(\mathcal{C})$ can be written equivalently as
\begin{align*}
Q_{\mathrm{RA}}(\mathcal{C}) = \frac{1}{2m_w} \sum_{C \in \mathcal{C}} \bigl( w(C) - p(C) \bigr)
\end{align*}

Intuitively, $Q_{\mathrm{RA}}(\mathcal{C})$ measures the surplus of observed affinity against the calibrated baseline. As a result, this formulation favours intra-community affinity that exceeds a baseline determined by node degrees and representation distance.

\spara{Comparison with a naive representation-aware objective.}
A natural alternative is to replace graph adjacency by representation-induced affinity while keeping the degree-based expected term, scaled by a global average affinity. This leads to a naive representation-aware modularity:
\begin{align*}
Q_{\mathrm{naive}}(\mathcal{C}) = \frac{1}{2m} \sum_{C \in \mathcal{C}} \sum_{u,v \in C} \left[ A_{uv}\cdot\phi(z_u,z_v) - \frac{k_u k_v}{2m}\bar{\phi} \right]
\end{align*}
, where $\bar{\phi}$ denotes the global average affinity. Since the expected term is insensitive to representation distance, this objective may reward high-affinity edges between nearby representations even when such affinity is already expected because the two endpoints are close in representation space.

In contrast, {\RAM} replaces the global average affinity $\bar{\phi}$ with the calibration function $g(\cdot)$. The difference is therefore not merely whether representations are used, but where they enter the observed-minus-expected principle: {\RAM} uses representations in both terms, defining the observed affinity and calibrating the expected affinity. This allows the objective to reward affinity that exceeds what is expected at the corresponding representation distance.

The following property shows that {\RAM} recovers classical modularity as a special case.

\begin{property}[Reduction to Classical Modularity]\label{prop:reduction}
If $\phi(z_u, z_v)=1$ for all $(u,v)\in E$ and $g(d)=1$ for all distances $d$, then {\RAM} reduces to classical modularity with resolution parameter $\gamma$.\end{property}
\begin{proof}

Under $\phi(z_u, z_v)=1$ for all $(u,v)\in E$, we have $w_{uv}=A_{uv}$ for all $u, v \in V$. Consequently, the total affinity mass aligns with the structural edge mass: $2m_w = \sum_{u,v \in V} w_{uv} = \sum_{u,v \in V} A_{uv} = 2m$. Moreover, under $g(d)=1$, the distance-calibrated null term simply becomes the degree-based expectation $\gamma \frac{k_u k_v}{2m}$. Substituting these directly into the definition of $Q_{\mathrm{RA}}(\mathcal{C})$ yields
\begin{align*}
Q_{\mathrm{RA}}(\mathcal{C}) = \frac{1}{2m} \sum_{C \in \mathcal{C}} \sum_{u,v \in C} \left( A_{uv} - \gamma \frac{k_u k_v}{2m} \right)
\end{align*}
, which is the classical modularity with resolution parameter $\gamma$.
\end{proof}
\section{PROBLEM STATEMENT}

Having defined the {\RAM} objective function, we formalise the associated optimisation problem. Given a graph and its node representations, our goal is to find a partition that maximises {\RAM}. As in classical modularity, the number of communities is not specified in advance, but naturally emerges from the optimisation process.

\begin{problemDefinition}[Representation-Aware Modularity Maximisation]\label{problem:RAMM}
Given an undirected graph $G=(V,E)$ and node representations $\mathcal{Z}$, the objective is to find an optimal partition $\mathcal{C}^*=\{C_1, \dots, C_k\}$ of $V$ such that $\mathcal{C}^* = \arg\max_{\mathcal{C}} \; Q_{\mathrm{RA}}(\mathcal{C})$. The partition must satisfy $C_i \cap C_j = \emptyset$ for $i \neq j$ and $\bigcup_{i=1}^k C_i = V$.
\end{problemDefinition}

Note that the number of communities $k$ is not predefined and emerges from optimisation.

\begin{example}
\label{ex:ram}
Consider the toy graph in Figure~\ref{fig:overview}. We compare $\mathcal{C}_1=\{\{1,\ldots,5\},\{6,7,8,9\}\}$ and $\mathcal{C}_2=\{\{1,\ldots,6\},\{7,8,9\}\}$. With $\gamma=1$, cosine similarity, and $\hat{g}$ estimated from all $\binom{9}{2}=36$ pairs using $L=5$ equal-width bins, {\RAM} gives $Q_{\mathrm{RA}}(\mathcal{C}_1)\approx 0.470$ and $Q_{\mathrm{RA}}(\mathcal{C}_2)\approx 0.391$. Although node~$6$ is near the first group in the representation space, adding it creates several intra-community non-edge pairs, such as $(1,6)$, $(2,6)$, and $(3,6)$. These pairs contribute nothing to the observed affinity $w(C)$ because affinity is graph-supported, but they still increase the null term $p(C)$. Hence, the surplus $w(C)-p(C)$ decreases, and {\RAM} favours the partition with stronger graph-supported affinity over the distance-calibrated baseline.
\end{example}

We next analyse two theoretical aspects of {\RAM}: the computational complexity of its maximisation problem, and the stability of the objective under estimation error of the calibration function $g(\cdot)$.

\begin{theorem}\label{thm:np}
Representation-aware modularity maximisation problem is $NP$-hard.
\end{theorem}

\begin{proof}
We prove this by reduction from classical modularity maximisation, which is known to be $NP$-hard~\cite{brandes2006maximizing}. As established in Property~\ref{prop:reduction}, setting $\phi(z_u, z_v) = 1$, $g(d) = 1$, and $\gamma=1$ reduces {\RAM} exactly to the classical modularity formulation. Since classical modularity is a special case of $Q_{\mathrm{RA}}$, it follows that maximising $Q_{\mathrm{RA}}(\mathcal{C})$ is at least as hard, confirming its $NP$-hardness.
\end{proof}

\begin{theorem}[Inapproximability of {\RAM} maximisation]\label{thm:inappro}
For any constant $\rho > 0$, unless $P=NP$, there is no polynomial-time algorithm that returns, for every instance with $Q_{\mathrm{RA}}^* > 0$, a partition $\mathcal{C}$ satisfying $Q_{\mathrm{RA}}(\mathcal{C}) \ge \rho \cdot Q_{\mathrm{RA}}^*$, where $Q_{\mathrm{RA}}^*$ denotes the optimal maximum value of {\RAM}.
\end{theorem}

\begin{proof}
Since classical modularity maximisation is a special case of {\RAM} by Property~\ref{prop:reduction}, any polynomial-time $\rho$-approximation algorithm for {\RAM} yields a polynomial-time $\rho$-approximation algorithm for classical modularity maximisation, contradicting its inapproximability for any constant $\rho>0$ unless $P=NP$~\cite{dinh2015network}.
\end{proof}

\begin{proposition}[Stability under estimation error of $g(d)$]
Let $Q_{\mathrm{RA}}^{(g)}(\mathcal{C})$ and $Q_{\mathrm{RA}}^{(\hat{g})}(\mathcal{C})$ denote the {\RAM} objectives computed with $g$ and $\hat{g}$, respectively. If $\sup_d |\hat{g}(d)-g(d)| \le \varepsilon$, then for any partition $\mathcal{C}=\{C_1,\dots,C_k\}$,
\begin{align*}
\bigl|Q_{\mathrm{RA}}^{(\hat{g})}(\mathcal{C})-Q_{\mathrm{RA}}^{(g)}(\mathcal{C})\bigr| \le \frac{\gamma\varepsilon}{4m\,m_w}\sum_{i=1}^k \mathrm{vol}(C_i)^2
\end{align*}
, where $\mathrm{vol}(C_i)=\sum_{u\in C_i} k_u$. In particular,
\begin{align*}
\bigl|Q_{\mathrm{RA}}^{(\hat{g})}(\mathcal{C})-Q_{\mathrm{RA}}^{(g)}(\mathcal{C})\bigr| \le \gamma\varepsilon \cdot \frac{m}{m_w}
\end{align*}
\end{proposition}

\begin{proof}
From the definition of {\RAM},
\begin{multline*}
Q_{\mathrm{RA}}^{(\hat{g})}(\mathcal{C}) - Q_{\mathrm{RA}}^{(g)}(\mathcal{C}) \\
= \frac{\gamma}{4m\,m_w} \sum_{i=1}^k \sum_{u,v \in C_i} k_u k_v \bigl(g(d_{uv}) - \hat{g}(d_{uv})\bigr)
\end{multline*}
Taking absolute values and using $|\hat{g}(d)-g(d)|\le \varepsilon$ gives
\begin{align*}
\bigl|Q_{\mathrm{RA}}^{(\hat{g})}(\mathcal{C})-Q_{\mathrm{RA}}^{(g)}(\mathcal{C})\bigr| \le \frac{\gamma\varepsilon}{4m\,m_w}\sum_{i=1}^k \sum_{u,v\in C_i} k_u k_v
\end{align*}
Since $\sum_{u,v\in C_i} k_u k_v = \bigl(\sum_{u\in C_i} k_u\bigr)^2 = \mathrm{vol}(C_i)^2$, the first bound follows. Also, $\sum_{i=1}^k \mathrm{vol}(C_i)^2 \le \bigl(\sum_{i=1}^k \mathrm{vol}(C_i)\bigr)^2 = (2m)^2$, which yields
\begin{align*}
\bigl|Q_{\mathrm{RA}}^{(\hat{g})}(\mathcal{C})-Q_{\mathrm{RA}}^{(g)}(\mathcal{C})\bigr| \le \gamma\varepsilon \cdot \frac{m}{m_w} 
\end{align*}
Therefore, the proposition is proved.
\end{proof}
\section{ALGORITHMS}

Maximising {\RAM} is challenging because its distance-calibrated null term depends on dense representation-distance interactions. We present {\BGD} as a divisive reference optimiser and {\MGA} as a scalable multi-level greedy optimiser with {\RAM}-specific gain updates and centroid summaries.

\subsection{Betweenness-Guided Divisive Search Algorithm (\BGD)}

As a reference optimiser for small and medium-sized graphs, we also consider a betweenness-guided divisive strategy inspired by Girvan--Newman~\cite{newman2004finding}. {\BGD} maintains a pruned working graph while evaluating every candidate partition using {\RAM} on the original graph. Starting from the original graph, it repeatedly removes edges with high weighted betweenness. To incorporate representation-aware affinity into the decomposition process, we define the length of an observed edge $(u,v)$ as $\ell_{uv}=1/(\epsilon+w_{uv})$, where $\epsilon>0$ is a small smoothing constant. Thus, paths are biased towards high-affinity edges, and low-affinity bridge edges are more likely to receive high betweenness and be removed earlier.

Whenever an edge removal splits a connected component, the connected components of the current working graph define a candidate partition. {\BGD} keeps the best partition found over the sequence of splits and returns the one with the largest {\RAM} value. The process continues until no further splits are possible. As a practical search convention, {\BGD} reports the best non-trivial partition encountered after at least one split. The all-in-one partition is therefore not used as the returned solution, although it remains a valid partition under the objective definition.

The dominant cost of {\BGD} is repeated weighted edge-betweenness computation. Exact Brandes computation costs $O(nm+n^2\log n)$ per recomputation on a weighted graph. We approximate it by sampling $x$ source nodes, reducing the cost to $O(x(m+n\log n))$ per recomputation; by default, $x=\sqrt n$. Since many recomputations may be needed before decomposition, {\BGD} remains costly and is used mainly as a reference optimiser. Scalable optimisation on large graphs is handled by {\MGA}.


\subsection{Multi-Level Greedy Aggregation Algorithm (\MGA)}
\label{sec:MGA}

\spara{Framework overview.}
To optimise {\RAM} at scale, we propose {\MGA}, a multi-level hierarchical optimisation framework. At each level, the {\single} procedure visits each node and moves it to a neighbouring community whenever the move gain is positive. The resulting communities are then contracted into super-nodes to form a coarsened graph. This cycle of local refinement and graph contraction then repeats until no merge occurs.

\subsubsection{{\single}}\label{sec:single}\mbox{}\\
The {\single} performs a single greedy pass: each node is moved at most once if a neighbouring community yields a positive gain. The pass terminates when no further positive move exists.

\SetAlgoSkip{}
\begin{algorithm}[t]
\caption{SingleLevelGreedy($G, \mathcal{Z}, \phi, \gamma, \hat{g}$)}
\label{alg:singlelevel}
\small
\SetAlgoLined
\DontPrintSemicolon
\SetKwFunction{CA}{computeAffinities}
\SetKwFunction{EC}{estimateCalibration}
\SetKwFunction{IS}{initSummaries}
\SetKwFunction{MV}{move}
\SetKwFunction{US}{updateSummaries}
$\mathcal{C}^\star \leftarrow \{\{v\} : v \in V\}$\;
$\IS(\mathcal{C}^\star)$\;
$\mathit{visited}[v] \leftarrow \mathrm{false}$ for all $v \in V$\;
$\mathit{improvement} \leftarrow \mathrm{true}$\;
\While{\textit{improvement}}{
    $\mathit{improvement} \leftarrow \mathrm{false}$\;
    \ForEach{\textnormal{node} $u \in V$ \textnormal{with} $\mathit{visited}[u] = \mathrm{false}$}{
        $A \leftarrow C(u)$\;
        $\mathcal{N}(u) \leftarrow \{C(v) \mid v \in N(u)\}$\;
        $B \leftarrow \underset{C \in \mathcal{N}(u)}{\arg\max}\, \Delta Q(u: A \rightarrow C)$\;
        \If{$\Delta Q > 0$}{
            $\MV(u, A, B, \mathcal{C}^\star)$\;
            $\US(A, B)$\;
            $\mathit{visited}[u] \leftarrow \mathrm{true}$\;
            $\mathit{improvement} \leftarrow \mathrm{true}$\;
        }
    }
}
\Return{$\mathcal{C}^\star$}
\end{algorithm}

\spara{Local gain under {\RAM}.}
Consider moving a node $u$ from its current community $A$ to a candidate community $B$. Because $Q_{\mathrm{RA}}$ is defined as a sum over node pairs, this move only affects the observed affinity and distance-calibrated null terms directly involving $u$.

For the observed term, since $w_{uv}$ is zero for non-adjacent nodes, only neighbours of $u$ contribute: $w(u, C) = \sum_{v \in N(u)\cap C} w_{uv}$.
For the null term, all node pairs involving $u$ contribute: $p(u, C) = \sum_{v \in C} p_{uv}$.

Unlike classical modularity, where the gain depends only on the neighbours of $u$ and a community-level degree sum, the {\RAM} null term aggregates $g(d_{uv})$ over all pairs between $u$ and the community. As a result, the local gain combines the graph-supported observed affinity with the community-aggregated null term.
Combining these, the objective gain from moving $u$ is $\Delta Q(u: A \rightarrow B)
=$
\begin{equation*}
\scalebox{0.95}{$\displaystyle
\frac{1}{2m_w}
\Big(
\big[w(u,B) - p(u,B)\big]
-
\big[w(u,A\setminus\{u\}) - p(u,A\setminus\{u\})\big]
\Big)
$}
\end{equation*}

\spara{Local search strategy.}
To efficiently find the best move for a node $u$, we only consider candidate communities that contain at least one of its neighbours. This restricted candidate set is defined as $\mathcal{N}(u) = \{ C(v) \mid v \in N(u) \}$, where $C(v)$ denotes the community that node $v$ belongs to. For each candidate community $B \in \mathcal{N}(u)$, we compute the objective gain $\Delta Q(u: A \rightarrow B)$. We then move $u$ to the community that offers the highest positive gain ($\Delta Q > 0$). Otherwise, $u$ stays in its current community $A$.

\spara{Centroid-based approximation.}
Evaluating $p(u, C) = \sum_{v \in C} p_{uv}$ exactly requires iterating over all nodes in $C$, which makes the per-move cost scale with community size. To reduce this cost, we approximate $p(u,C)$ using a community-level summary. For each community $C$, we maintain the degree sum $K_C = \sum_{v \in C} k_v$, the embedding sum $Z_C = \sum_{v \in C} z_v$, and the cardinality $|C|$. The centroid of $C$ is $\bar z_C = Z_C / |C|$, and the null interaction between $u$ and $C$ is approximated as $\tilde p(u, C) = \gamma \frac{k_u K_C}{2m} \hat g\big(\|z_u - \bar z_C\|\big).$

This centroid-based approximation replaces $|C|$ distance-dependent evaluations with a single evaluation of $\hat g$ at the node--centroid distance, reducing the null-term cost from $O(|C|d)$ to $O(d)$ per candidate community.

Whenever a node $u$ moves from $A$ to $B$, the three summaries are updated incrementally:
\begin{align*}
&Z_A \leftarrow Z_A - z_u, \quad K_A \leftarrow K_A - k_u, \quad |A| \leftarrow |A| - 1,\\
&Z_B \leftarrow Z_B + z_u, \quad K_B \leftarrow K_B + k_u, \quad |B| \leftarrow |B| + 1.
\end{align*}
Each update runs in $O(d)$ time and requires $O(kd)$ total memory, where $k$ is the current number of communities. Combined with the local search over $\mathcal{N}(u)$, the cost of evaluating and executing a move becomes $O(|N(u)| + |\mathcal{N}(u)|d)$.

We bound the error introduced by this approximation. Since the approximation replaces multiple distance evaluations with a single one at the centroid, its accuracy depends on how $\hat g(\cdot)$ varies with distance and how spread out the community members are around their centroid. We make this dependence precise as follows.

First, assume that $\hat g(\cdot)$ is $L_{\hat g}$-Lipschitz, i.e., $|\hat g(x)-\hat g(y)| \le L_{\hat g}|x-y|$ for all $x,y \ge 0$. Second, we measure how spread out community $C$ is around its centroid via $\sigma_C = \frac{1}{K_C} \sum_{v \in C} k_v \|z_v-\bar z_C\|$, which is small when the members of $C$ are tightly clustered around $\bar z_C$.

Combining these, the approximation error satisfies
Combining these, the approximation error satisfies
\begin{align*}
\scalebox{0.8}{$\left|p(u,C)-\tilde p(u,C)\right|$}
&\scalebox{0.8}{$\displaystyle
{}=
\left|
\gamma \frac{k_u}{2m}
\sum_{v \in C} k_v \hat g(\|z_u-z_v\|)
-
\gamma \frac{k_u K_C}{2m} \hat g(\|z_u-\bar z_C\|)
\right|$} \\
&\scalebox{0.9}{$\displaystyle
{}\le
\gamma \frac{k_u}{2m}
\sum_{v \in C} k_v
\left|
\hat g(\|z_u-z_v\|)
-
\hat g(\|z_u-\bar z_C\|)
\right|$} \\
&\scalebox{0.9}{$\displaystyle
{}\le
\gamma \frac{k_u L_{\hat g}}{2m}
\sum_{v \in C} k_v
\left|
\|z_u-z_v\|-\|z_u-\bar z_C\|
\right|$} \\
&\scalebox{0.9}{$\displaystyle
{}\le
\gamma \frac{k_u L_{\hat g}}{2m}
\sum_{v \in C} k_v\|z_v-\bar z_C\|
=
\gamma \frac{k_u K_C}{2m} L_{\hat g}\sigma_C$}
\end{align*}
The final bound has a clear interpretation: the approximation error is small when either $\hat g(\cdot)$ varies slowly with distance (small $L_{\hat g}$) or the community members are tightly clustered in the representation space (small $\sigma_C$). The latter condition is naturally aligned with {\RAM}, which favours communities whose members are coherent in the representation space. This makes the centroid approximation a reasonable choice within the {\RAM} framework.

\spara{Termination.}
Each accepted move has a positive estimated gain $\Delta Q>0$. To avoid cyclic reassignments under the centroid-based approximation, we use a locked greedy pass in which each node is moved at most once. This caps the number of accepted moves by $n$, and the procedure terminates when no unlocked node has a positive-gain move to a neighbouring community. The {\single} procedure is summarised in Algorithm~\ref{alg:singlelevel}.

\subsubsection{Aggregation under {\RAM}}\label{sec:aggregation}\mbox{}\\
A key requirement in the multi-level design is that {\RAM} remains consistent under graph contraction. Let $X$ and $Y$ denote two super-nodes, each representing a set of original nodes. We define the aggregated observed affinity and distance-calibrated null terms as
\begin{align*}
W_{XY} = \sum_{u \in X} \sum_{v \in Y} w_{uv}, \quad \text{and} \quad
P_{XY} = \sum_{u \in X} \sum_{v \in Y} p_{uv}
\end{align*}
Intuitively, $W_{XY}$ sums the observed affinity of the original edges between $X$ and $Y$, while $P_{XY}$ sums the distance-calibrated null contributions over all node pairs with one endpoint in $X$ and one in $Y$. If both quantities are computed exactly, evaluating $Q_{\mathrm{RA}}$ on the coarsened graph is algebraically equivalent to evaluating it on the original graph for the corresponding partition. This exact aggregation principle clarifies what must be preserved under contraction. In our scalable implementation, however, storing and updating all exact pairwise null interactions $P_{XY}$ can be computationally expensive. Therefore, {\MGA} uses the centroid-based summaries introduced in Section~\ref{sec:single} to approximate the null interactions at coarser levels. 

\subsubsection{Multi-Level Iteration}\label{sec:multi-level-iteration}\mbox{}\\
\textbf{Local gain at the super-node level.}
At each level, {\MGA} applies the same gain formulation as the {\single}, now at the super-node level. For a super-node $X$ and a community $C$, we define the aggregated observed and distance-calibrated null terms as
\begin{align*}
W(X, C) = \sum_{Y \in C} W_{XY}, \quad \text{and} \quad
P(X, C) = \sum_{Y \in C} P_{XY}
\end{align*}
Using these, the gain from moving $X$ from its current community $A$ to a candidate community $B$ is given by $\Delta Q(X: A \rightarrow B)
=$
\begin{equation*}
\scalebox{0.87}{$\displaystyle
\frac{1}{2m_w}
\Big(
[W(X,B) - P(X,B)]
-
[W(X,A\setminus\{X\}) - P(X,A\setminus\{X\})]
\Big)
$}
\end{equation*}

\begin{algorithm}[t]
\caption{{\MGA}($G, \mathcal{Z}, \phi, \gamma$)}
\label{alg:mga}
\small
\SetAlgoLined
\DontPrintSemicolon
\SetKwFunction{CA}{computeAffinities}
\SetKwFunction{EC}{estimateCalibration}
\SetKwFunction{SLM}{SingleLevelGreedy}
\SetKwFunction{CT}{graphContract}
\SetKwFunction{PJ}{projectToOriginal}

$\CA(G, \mathcal{Z}, \phi)$\;
$\hat{g} \leftarrow \EC(G, \mathcal{Z}, \phi)$\;
$G' \leftarrow G$\;
\While{\textnormal{true}}{
    $\mathcal{C} \leftarrow \SLM(G',\mathcal{Z},\phi,\gamma,\hat{g})$\;
    \If{$|\mathcal{C}| = |V'|$}{
        \textbf{break}\;
    }
    $G' \leftarrow \CT(G', \mathcal{C})$\;
}
$\mathcal{C}^\star \leftarrow \PJ(\mathcal{C})$\;
\Return{$\mathcal{C}^\star$}
\end{algorithm}

\spara{Algorithm description.}
At each level $t$, we apply {\single} to the current graph $G^{(t)}$ to obtain a partition $\mathcal{C}^{(t)}$. If no merges occur at this level, the algorithm terminates. Otherwise, each community in $\mathcal{C}^{(t)}$ is contracted into a single super-node, yielding the coarsened graph $G^{(t+1)}$ with aggregated observed affinities $W_{XY}$ and per-super-node summaries $(Z_X, K_X, |X|)$. These summaries are used to approximate the distance-calibrated null interactions at the coarser level, rather than explicitly materialising all exact pairwise terms $P_{XY}$. The procedure then continues on $G^{(t+1)}$ until termination, at which point the final super-node partition is mapped back to the original nodes. Note that, although we describe the procedure using indexed graphs $G^{(0)}, G^{(1)}, \ldots$ for clarity, our implementation keeps only the current-level graph $G'=(V',E')$ in memory. The procedure is given in Algorithm~\ref{alg:mga}.

\spara{Design rationale.}
{\MGA} combines local refinement with multi-level aggregation. The {\single} improves the partition through node-level reassignments, but local moves alone can trap the search in poor local optima and provide only a limited view of global structure. Graph contraction mitigates this limitation by enabling larger structural changes across coarsened levels. Unlike Louvain~\cite{blondel2008fast}, however, aggregation in {\MGA} must account for both observed affinity and distance-calibrated null interactions: while Louvain tracks only aggregated edge weights because its null model is degree-based, {\MGA} additionally maintains per-super-node representation summaries to approximate the aggregated null interactions at each level.

\spara{Time complexity.}
Let $L$ be the number of coarsening levels, and let $(n_t,m_t)$ denote the number of nodes and edges at level $t$, and let $\alpha_t$ denote the number of {\single} passes at level $t$. With the centroid-based approximation, evaluating and executing a move for node $u$ costs $O(|N(u)|+|\mathcal{N}(u)|d)$ by scanning its neighbours once, where $\mathcal{N}(u)$ is the set of neighbouring communities. Since $|\mathcal{N}(u)|\le |N(u)|$, one pass of {\single} at level $t$ costs $O(m_t d)$, and $\alpha_t$ passes cost $O(\alpha_t m_t d)$. Graph contraction, including the update of centroid summaries, costs $O(m_t+n_t d)$. Including the preprocessing cost $O(md+|\mathcal{S}|d)$ for edge affinities and the sampled-pair calibration function, the total running time of {\MGA} is $O\!\left( md+|\mathcal{S}|d+ \sum_{t=0}^{L-1}\bigl(\alpha_t m_t d+m_t+n_t d\bigr) \right)$.  In practice, $\alpha_t$ is small and $(n_t,m_t)$ decreases rapidly after coarsening, so the running time is dominated by the first few coarsening levels, which is consistent with the near-linear trend observed in EQ8.

\section{EXPERIMENTS}

\subsection{Experimental Questions}
To evaluate the proposed algorithm, we conduct extensive experiments to address the following eight experimental questions (\textbf{EQs}):

\begin{itemize}[leftmargin=*]
    \item \textbf{EQ1. Topology-based baselines:}
    How does the proposed model compare with topology-based community detection methods in effectiveness and efficiency?
    
    \item \textbf{EQ2. Embedding-based baselines:}
    How does the proposed model compare with clustering methods that use the same node representations?

    \item \textbf{EQ3. Deep learning-based baselines:}
    How does the proposed model compare with deep community detection methods, and can their learned embeddings further improve its overall effectiveness?
    
    \item \textbf{EQ4. Effect of $\gamma$:}
    How does the resolution parameter $\gamma$ affect community detection quality?
    
    \item \textbf{EQ5. Distance calibration:}
    Does removing $\hat{g}(\cdot)$, the degree factor, or representation--graph alignment degrade overall clustering performance?
    
    \item \textbf{EQ6. Centroid approximation:}
    How does centroid-based null-term approximation affect quality and running time?
    
    \item \textbf{EQ7. Robustness:}
    How robust is the model to different similarity functions and embedding methods?
    
    \item \textbf{EQ8. Scalability:}
    How does the algorithm scale in running time as the graph size increases?
\end{itemize}

\subsection{Experimental Setting}
\spara{Datasets.}
We use seven real-world datasets: Karate~\cite{zachary1977information}, Cora, Citeseer, and Amazon PC~\cite{fey2019fastgraphrepresentationlearning} have disjoint ground-truth communities, while DBLP, Amazon, and YouTube from SNAP~\cite{snapnets} contain overlapping ground-truth communities. The datasets cover both small-scale attributed networks and large-scale real-world networks. Table~\ref{tab:datasets} summarises their basic statistics, including the number of nodes ($\boldsymbol{|V|}$), the number of edges ($\boldsymbol{|E|}$), the number of clusters ($\boldsymbol{|C|}$), and the average degree ($\boldsymbol{d_{avg}}$).

\begin{table}[t]
\caption{Real-world datasets}
\label{tab:datasets}
\centering
\small
\setlength{\tabcolsep}{4pt}
\renewcommand{\arraystretch}{1.1}

\begin{tabular}{c|c|c|c|c|c}
\toprule
\textbf{Type} & \textbf{Dataset}
& $\boldsymbol{|V|}$
& $\boldsymbol{|E|}$
& $\boldsymbol{|C|}$
& $\boldsymbol{d_{avg}}$ \\
\midrule
\multirow{4}{*}{Small-scale}
& Karate    & 34        & 78        & 2      & 4.6 \\
& Cora      & 2,708     & 5,278    & 7      & 3.9 \\
& Citeseer  & 3,327     & 4,552     & 6      & 2.7 \\
& Amazon PC & 13,752    & 245,861   & 10     & 35.8 \\
\midrule
\multirow{3}{*}{Large-scale}
& DBLP$^{*}$    & 317,080   & 1,049,866 & 13,477 & 6.6 \\
& Amazon$^{*}$  & 334,863   & 925,872   & 75,149 & 5.5 \\
& YouTube$^{*}$ & 1,134,890 & 2,987,624 & 8,385  & 5.3 \\
\bottomrule
\end{tabular}

\vspace{2pt}
\parbox{\linewidth}{\footnotesize\raggedright
$^{*}$ Dataset with overlapping ground-truth communities.}
\vspace{-0.6cm}

\end{table}

\begin{table*}[t]
\centering
\caption{EQ2. Performance comparison on small-scale datasets (AG+: baseline with AutoGL embeddings (same as {\MGA}))}
\label{tab:small_comparison}
\small
\resizebox{\textwidth}{!}{
\begin{tabular}{l | ccccc | ccccc | ccccc}
\toprule
 & \multicolumn{5}{c|}{\textbf{NMI} ($\uparrow$)} & \multicolumn{5}{c|}{\textbf{ARI} ($\uparrow$)} & \multicolumn{5}{c}{\textbf{Time (s)} ($\downarrow$)} \\
\cmidrule(lr){2-6} \cmidrule(lr){7-11} \cmidrule(lr){12-16}
\textbf{Dataset} & \textbf{{\MGA}} & \textbf{AG+Kms} & \textbf{AG+Spec} & \textbf{AG+OPTICS} & \textbf{AG+HDB} & \textbf{{\MGA}} & \textbf{AG+Kms} & \textbf{AG+Spec} & \textbf{AG+OPTICS} & \textbf{AG+HDB} & \textbf{{\MGA}} & \textbf{AG+Kms} & \textbf{AG+Spec} & \textbf{AG+OPTICS} & \textbf{AG+HDB} \\
\midrule
Karate & \textbf{0.8519} & 0.2260 & \underline{0.8255} & 0.5434 & 0.2260 & \textbf{0.8289} & 0.1095 & \underline{0.8027} & 0.5019 & 0.1095 & \textbf{0.0006} & 0.0920 & 0.0309 & 0.0323 & \underline{0.0010} \\
Cora & \underline{0.4673} & \textbf{0.5131} & 0.0249 & 0.3799 & 0.2127 & \underline{0.2808} & \textbf{0.3730} & 0.0020 & 0.0025 & -0.0106 & \underline{0.7452} & 2.3741 & 3.4009 & 1.0971 & \textbf{0.5323} \\
Citeseer & 0.3351 & 0.1425 & 0.0758 & \underline{0.3563} & \textbf{0.3593} & \textbf{0.1232} & \underline{0.0410} & 0.0008 & 0.0032 & 0.0131 & \textbf{1.2368} & 5.9057 & 6.8476 & 2.1494 & \underline{1.8969} \\
Amazon PC & \textbf{0.5523} & 0.2021 & 0.0033 & \underline{0.3298} & 0.0495 & \textbf{0.4049} & \underline{0.1554} & -0.0005 & 0.0000 & 0.0000 & \textbf{2.4663} & 9.3387 & 318.0584 & 15.4501 & \underline{8.9653} \\
\bottomrule
\end{tabular}
}
\end{table*}

\begin{table*}[t]
\centering
\caption{EQ2. Performance comparison on large-scale datasets (AG+: baseline with AutoGL embeddings (same as {\MGA}))}
\label{tab:large_comparison}
\small
\resizebox{\textwidth}{!}{
\begin{tabular}{l | ccccc | ccccc | ccccc}
\toprule
 & \multicolumn{5}{c|}{\textbf{1 - Conductance} ($\uparrow$)} & \multicolumn{5}{c|}{\textbf{Modularity} ($\uparrow$)} & \multicolumn{5}{c}{\textbf{Time (s)} ($\downarrow$)} \\
\cmidrule(lr){2-6} \cmidrule(lr){7-11} \cmidrule(lr){12-16}
\textbf{Dataset} & \textbf{{\MGA}} & \textbf{AG+Kms} & \textbf{AG+Spec} & \textbf{AG+OPTICS} & \textbf{AG+HDB} & \textbf{{\MGA}} & \textbf{AG+Kms} & \textbf{AG+Spec} & \textbf{AG+OPTICS} & \textbf{AG+HDB} & \textbf{{\MGA}} & \textbf{AG+Kms} & \textbf{AG+Spec} & \textbf{AG+OPTICS} & \textbf{AG+HDB} \\
\midrule
Amazon & \textbf{0.9669} & \underline{0.8925} & OOT & 0.1882 & 0.3466 & \textbf{0.9214} & 0.3725 & OOT & 0.2269 & \underline{0.4325} & \textbf{6.7645} & \underline{284.2898} & OOT & 12408.9564 & 19545.6532 \\
DBLP & \textbf{0.8915} & \underline{0.7558} & OOT & 0.1320 & 0.2728 & \textbf{0.8091} & \underline{0.6808} & OOT & 0.1540 & 0.2862 & \textbf{4.6996} & \underline{129.0685} & OOT & 5664.4762 & 6210.7943 \\
YouTube & \textbf{0.9060} & \underline{0.6857} & OOT & OOT & OOT & \textbf{0.6969} & \underline{0.0823} & OOT & OOT & OOT & \textbf{20.0166} & \underline{685.6667} & OOT & OOT & OOT \\
\bottomrule
\end{tabular}
}
\end{table*}

\spara{Baselines.}
We compare {\MGA} with three groups of baselines. The topology-based baselines mainly consist of Girvan--Newman~\cite{newman2004finding}, Greedy Modularity~\cite{clauset2004finding}, Louvain~\cite{blondel2008fast}, Leiden~\cite{traag2019louvain}, and Infomap~\cite{rosvall2008maps}, covering classical divisive, agglomerative, modularity-based, and information-theoretic community detection methods. This selection follows a comparative analysis~\cite{lancichinetti2009community} reporting that modularity-based methods are among the top performers. The embedding-based baselines include $k$-means~\cite{mcqueen1967some}, spectral clustering~\cite{ng2001spectral}, OPTICS~\cite{ankerst1999optics}, and HDBSCAN~\cite{campello2013density}, all applied to the same node representations used by {\MGA}. We select the number of clusters for $k$-means via silhouette score~\cite{rousseeuw1987silhouettes} and for spectral clustering via the eigengap heuristic~\cite{von2007tutorial}, and set the minimum samples for OPTICS to $\ln n$~\cite{birant2007st}. The deep learning-based baselines include DMoN~\cite{tsitsulin2023graph} and DGCluster~\cite{bhowmick2024dgcluster}; we also evaluate DMoN+{\MGA} and DGC+{\MGA} by feeding their learned embeddings into {\MGA}. For deep baselines, we use the authors' official implementations with default hyperparameters. Methods that do not finish within 24 hours are marked as OOT.

\spara{Experimental setup.}
For the similarity function $\phi$, we use a rescaled cosine similarity $\phi(z_u, z_v) = (1 + \cos(z_u, z_v))/2$, mapping the cosine similarity to $[0, 1]$. The resolution parameter is set to $\gamma = 1$.
For the calibration function $g(d)$, we estimate it from sampled node pairs. Specifically, we sample 30\% of all possible node pairs uniformly at random, capped at a maximum of 5M pairs due to memory constraints. We then partition the sampled pairs by clustering their pairwise distances using $k$-means~\cite{mcqueen1967some}, where the number of clusters is automatically selected by maximising the silhouette score~\cite{rousseeuw1987silhouettes}. For datasets without node features, we initialise node representations using DeepWalk~\cite{perozzi2014deepwalk}. To ensure a fair comparison across datasets, we use AutoGL~\cite{zhang2021autogl} with its default settings to automatically select unsupervised node embedding methods, and use the recommended embeddings as input to our model. {\BGD} samples $\min(\sqrt{n}, 100)$ source nodes when estimating edge betweenness.

\spara{Evaluation metrics.}
For datasets with disjoint ground-truth communities, we report NMI~\cite{danon2005comparing} and ARI~\cite{hubert1985comparing}. For Amazon, DBLP, and YouTube, whose ground-truth communities overlap, we report conductance~\cite{kannan2004clusterings} and modularity~\cite{newman2004finding}. All scores are averaged over 10 independent runs.

\subsection{Experimental Result}

\begin{figure}[h]
    \centering
    \includegraphics[width=0.99\linewidth]{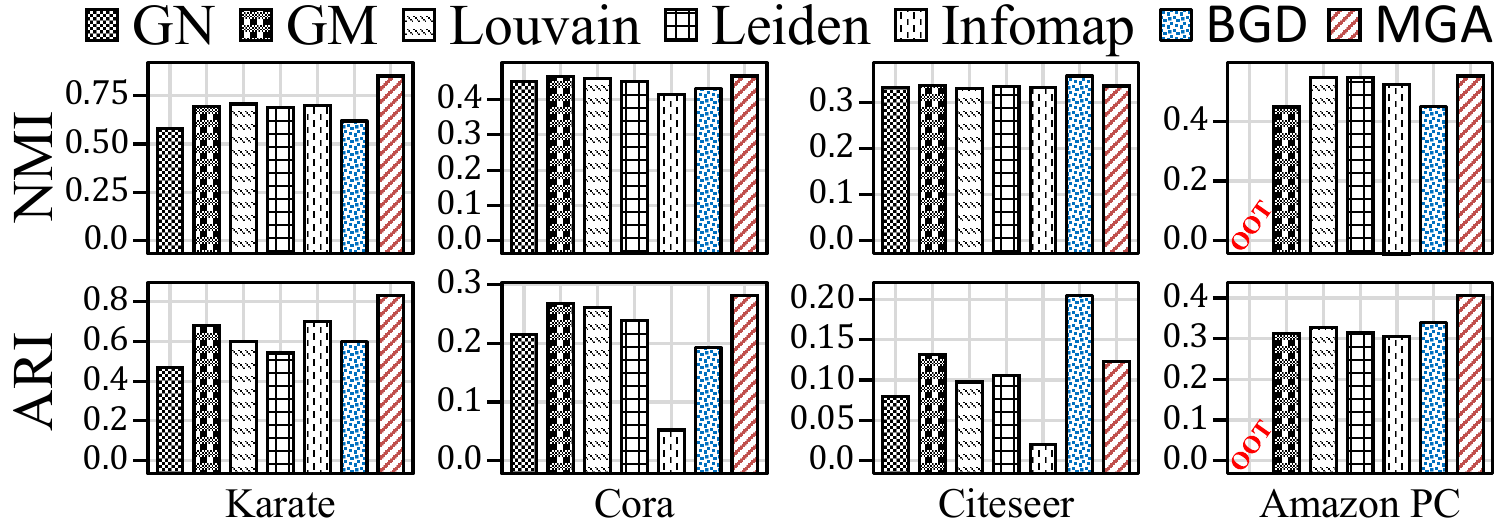}
    \vspace{-0.3cm}
    \caption{EQ1.1. Effectiveness on small-scale datasets}
    \label{fig:effect_small}
    \vspace{-0.4cm}
\end{figure}

\spara{EQ1.1. Effectiveness on small-scale datasets.}
Figure~\ref{fig:effect_small} reports the NMI and ARI results on the four real-world datasets with disjoint ground-truth communities. GN did not finish on Amazon~PC within the 24-hour time limit and is therefore marked as OOT. Overall, {\MGA} achieves higher clustering quality than the topology-only baselines across the evaluated datasets, improving NMI and ARI by $1.05\times$ and $1.10\times$ on average, respectively. The improvement is most pronounced on Karate, where {\MGA} obtains an NMI of $0.852$ and an ARI of $0.829$, exceeding the best baseline by $\boldsymbol{20.5}\%$ and $\boldsymbol{18.0}\%$, respectively. On Citeseer, all methods exhibit relatively low performance, which is attributable to its fragmented graph structure: the graph contains 390 connected components, and only 64.7\% of nodes belong to the largest component. Nevertheless, {\BGD} and {\MGA} remain competitive with, or outperform, the strongest baselines. These results indicate that representation-aware modularity can improve community alignment over topology-only objectives when node representations provide informative proximity signals.

\vspace{-0.2cm}
\begin{figure}[h]
    \centering
    \includegraphics[width=0.99\linewidth]{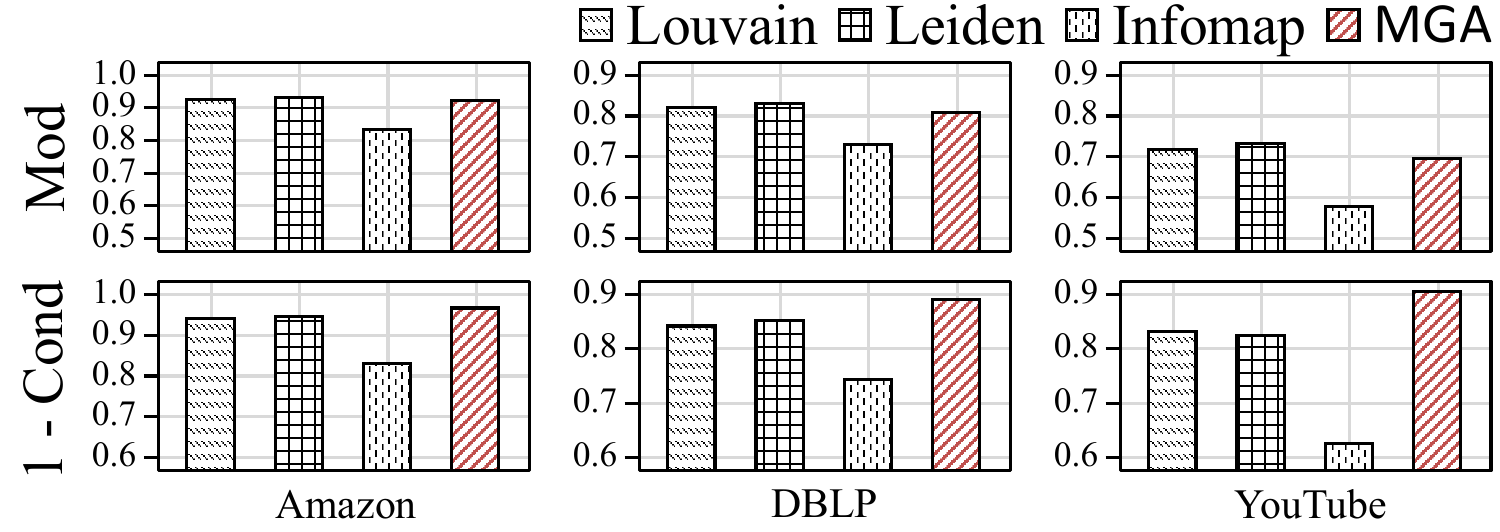}
    \vspace{-0.3cm}
    \caption{EQ1.2. Effectiveness on large-scale datasets (Mod: Modularity, 1 - Cond: $1-\text{conductance}$).}
    \label{fig:effect_large}
    \vspace{-0.4cm}
\end{figure}

\spara{EQ1.2. Effectiveness on large-scale datasets.}
Figure~\ref{fig:effect_large} reports modularity and $1-\text{conductance}$ on the three large-scale datasets with overlapping ground-truth communities. Since the ground-truth overlaps whereas the evaluated methods return non-overlapping partitions, NMI and ARI are not directly applicable. We report conductance~\cite{kannan2004clusterings}, which measures the normalised fraction of edges crossing community boundaries, and modularity, the canonical objective optimised by Louvain and Leiden. For consistency, we use $1-\text{conductance}$ so that higher values indicate better separation. GN, GM, and {\BGD} are excluded because they did not finish within the 24-hour time limit. On modularity, although {\MGA} does not directly optimise classical modularity, it remains close to Leiden, the strongest modularity-optimising baseline, with an average gap of $0.022$ across the three datasets, while achieving higher modularity than Infomap. On $1-\text{conductance}$, {\MGA} consistently outperforms all baselines across the three datasets. The improvement is most pronounced on YouTube, where {\MGA} obtains $0.906$, exceeding the best baseline by $9.0\%$. These results indicate that {\MGA} remains competitive under classical modularity while producing partitions with favourable boundary separation on large-scale networks.


\begin{figure}[h]
    \centering
    \includegraphics[width=0.99\linewidth]{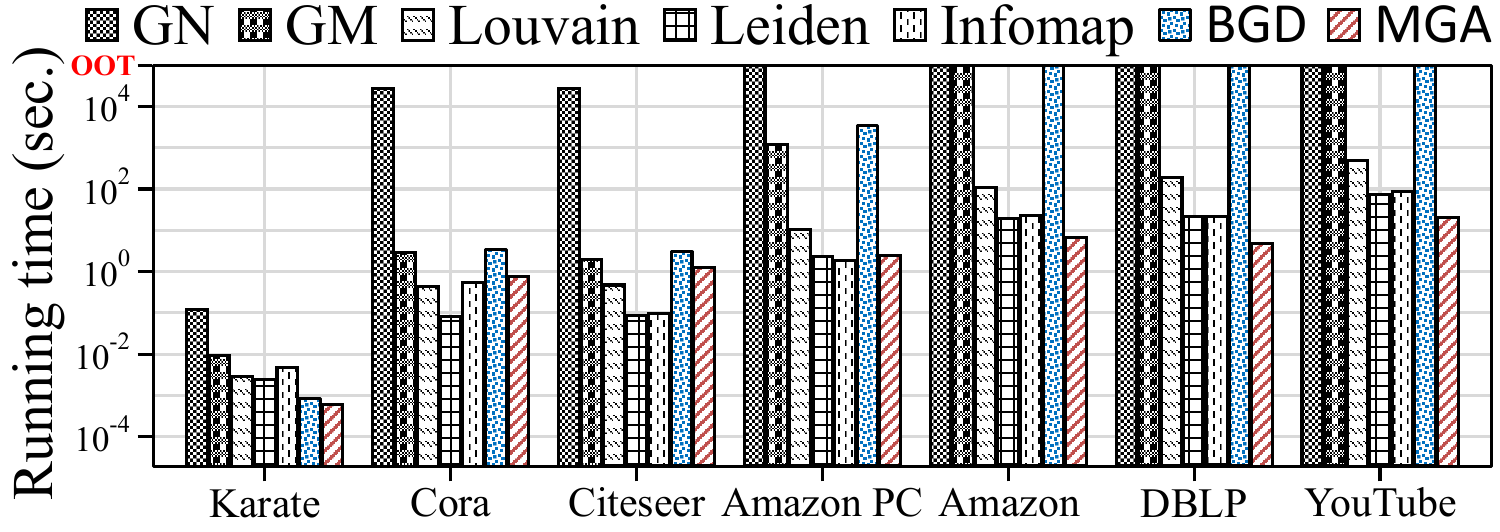}
    \vspace{-0.3cm}
    \caption{EQ1.3. Efficiency}
    \label{fig:efficiency}
    \vspace{-0.4cm}
\end{figure}

\spara{EQ1.3. Efficiency.} 
Figure~\ref{fig:efficiency} reports the total running time of all algorithms across the seven datasets. For {\MGA}, the reported time includes computing edge affinities $w_{uv}$, estimating $\hat{g}(d_{uv})$, and running the optimisation procedure; embedding learning time is excluded because node representations are treated as input to the proposed objective. On Cora and Citeseer, {\MGA} is slightly slower than Louvain, Leiden, and Infomap, mainly because the precomputation of $\hat{g}(d_{uv})$ and $w_{uv}$ accounts for over $98\%$ of its total running time. The optimisation step itself takes less than $0.02\,\mathrm{s}$, so the absolute overhead remains small. On the large-scale datasets, {\MGA} shows favourable running time under the same reporting protocol, running $3.69\times$ faster than Leiden on average. This behaviour is consistent with its multi-level design, where graph contraction reduces the size of subsequent optimisation instances. Overall, the results suggest that distance calibration incurs manageable additional preprocessing overhead and that {\MGA} remains practically scalable, although the reported times exclude embedding learning, which the topology-only baselines do not require.


\spara{EQ2. Embedding-based baselines.}
We compare {\MGA} with four embedding-based clustering baselines applied to the same node representations used as input to {\MGA}. Tables~\ref{tab:small_comparison} and~\ref{tab:large_comparison} report effectiveness and efficiency on the small-scale and large-scale datasets, respectively. On the small-scale datasets, {\MGA} achieves the best or second-best score in seven out of eight NMI/ARI cases. In contrast, every embedding-based baseline produces near-zero or negative ARI on at least one dataset, whereas {\MGA} maintains non-trivial ARI across all four datasets. The advantage of {\MGA} becomes more pronounced on the large-scale datasets. In terms of effectiveness, {\MGA} achieves the best $1-$conductance and modularity among the embedding-based baselines on all three large-scale datasets. In terms of efficiency, {\MGA} is also consistently the fastest on large graphs, achieving an average speedup of over $34\times$ over $k$-means and orders-of-magnitude improvements over OPTICS and HDBSCAN on the datasets where they finish. These results indicate that directly applying conventional clustering to node representations is not sufficient to obtain reliable community partitions. The improvement of {\MGA} is attributable not only to the use of embeddings, but also to {\RAM}, which couples representation geometry with graph topology through the distance-calibrated null model.

\begin{figure}[h]
    \centering
    \includegraphics[width=0.99\linewidth]{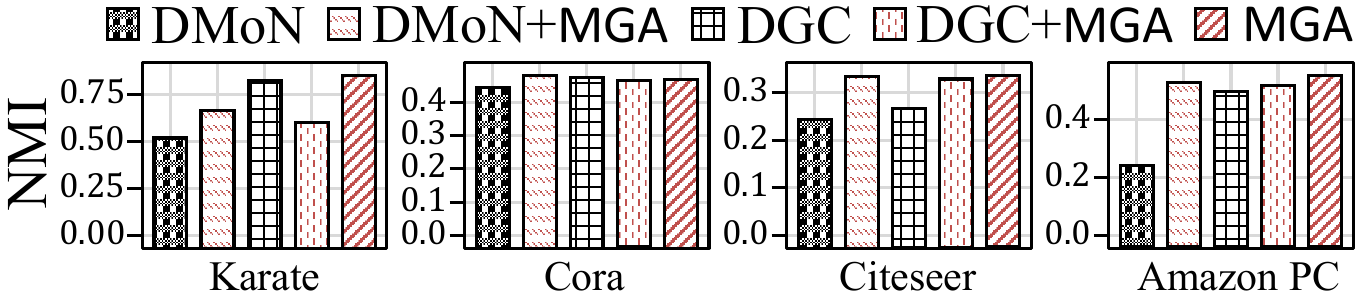}
    \vspace{-0.3cm}
    \caption{EQ3. Deep learning-based baselines}
    \label{fig:deep}
    \vspace{-0.4cm}
\end{figure}

\spara{EQ3. Deep learning-based baselines.}
Figure~\ref{fig:deep} reports the NMI of two deep learning-based community detection baselines and their {\MGA}-augmented variants on the four small-scale datasets. For each deep baseline, we conduct $10$ independent runs and feed the learned embeddings into {\MGA} to obtain the corresponding augmented variant. Across the four datasets, {\MGA} achieves the best or near-best NMI, showing that the proposed objective remains competitive with deep learning-based community detection methods. Moreover, applying {\MGA} on top of learned embeddings generally improves the corresponding baseline. For example, DMoN+{\MGA} outperforms DMoN on all four datasets, improving NMI by $1.48\times$ on average, with the largest gain on Amazon~PC, where NMI more than doubles from $0.243$ to $0.530$. These results suggest that {\MGA} can act as an effective representation-aware optimisation layer, rather than merely relying on the embedding quality itself.


\begin{figure}[h]
    \centering
    \includegraphics[width=0.99\linewidth]{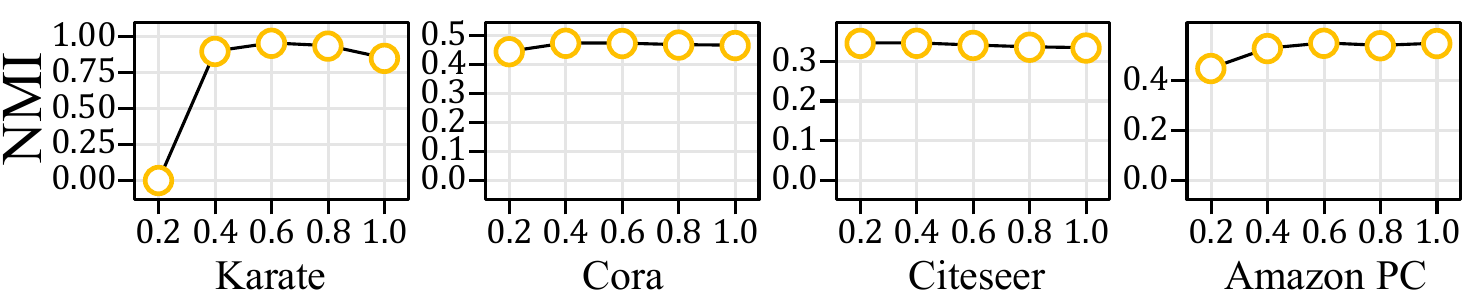}
    \vspace{-0.3cm}
    \caption{EQ4. Effect of $\gamma$}
    \label{fig:gamma}
    \vspace{-0.4cm}
\end{figure}

\spara{EQ4. Effect of $\gamma$.}
Figure~\ref{fig:gamma} reports the NMI of {\MGA} on the four small-scale datasets as the resolution parameter $\gamma$ varies from $0.2$ to $1.0$. On Cora, Citeseer, and Amazon~PC, the NMI increases in the small-$\gamma$ regime and then stabilises around $\gamma=0.4$, with $\gamma=1$ remaining within $1$--$3.5\%$ of the best observed NMI. Citeseer is particularly stable, with an NMI variation of less than $0.012$ over the entire range of $\gamma$. Karate shows a sharper response: the NMI peaks at $\gamma=0.6$ and decreases mildly toward $\gamma=1$. This is expected on such a small graph, where changes in the null term can more directly affect the resulting partition. Even at $\gamma=1$, however, {\MGA} obtains an NMI of $0.852$ on Karate, still outperforming all topology-only baselines in EQ1.1. Since classical modularity commonly adopts $\gamma=1$ as the default resolution setting, we use the same convention instead of tuning $\gamma$ separately for each dataset. These results show that {\MGA} is not highly sensitive to $\gamma$ around the default setting, supporting the use of $\gamma=1$ in the main experiments.


\begin{table}[h]
\centering
\caption{EQ5. Ablation variants of {\RAM}}
\label{tab:ablation_variants}
\small
\setlength{\tabcolsep}{6pt}
\renewcommand{\arraystretch}{1.3}
\begin{tabular}{l|c|c}
\toprule
\textbf{Variant} & \textbf{Observed term} & \textbf{Null term} \\
\midrule
{\RAM} w/o deg          & $A_{uv} \cdot \phi(z_u, z_v)$ & $\gamma\, \hat{g}(d_{uv})$ \\
{\RAM} w/o $\hat{g}$    & $A_{uv} \cdot \phi(z_u, z_v)$ & $\gamma\, \frac{k_u k_v}{2m}$ \\
{\RAM} w $\bar{\phi}$   & $A_{uv} \cdot \phi(z_u, z_v)$ & $\gamma\, \frac{k_u k_v}{2m}\, \bar{\phi}$ \\
{\RAM} w/o $\phi$       & $A_{uv}$                      & $\gamma\, \frac{k_u k_v}{2m}\, \hat{g}(d_{uv})$ \\
\midrule
\textbf{{\RAM} (full)}  & $A_{uv} \cdot \phi(z_u, z_v)$ & $\gamma\, \frac{k_u k_v}{2m}\, \hat{g}(d_{uv})$ \\
\bottomrule
\end{tabular}
\vspace{-0.2cm}
\end{table}

\begin{figure}[h]
    \centering
    \includegraphics[width=0.99\linewidth]{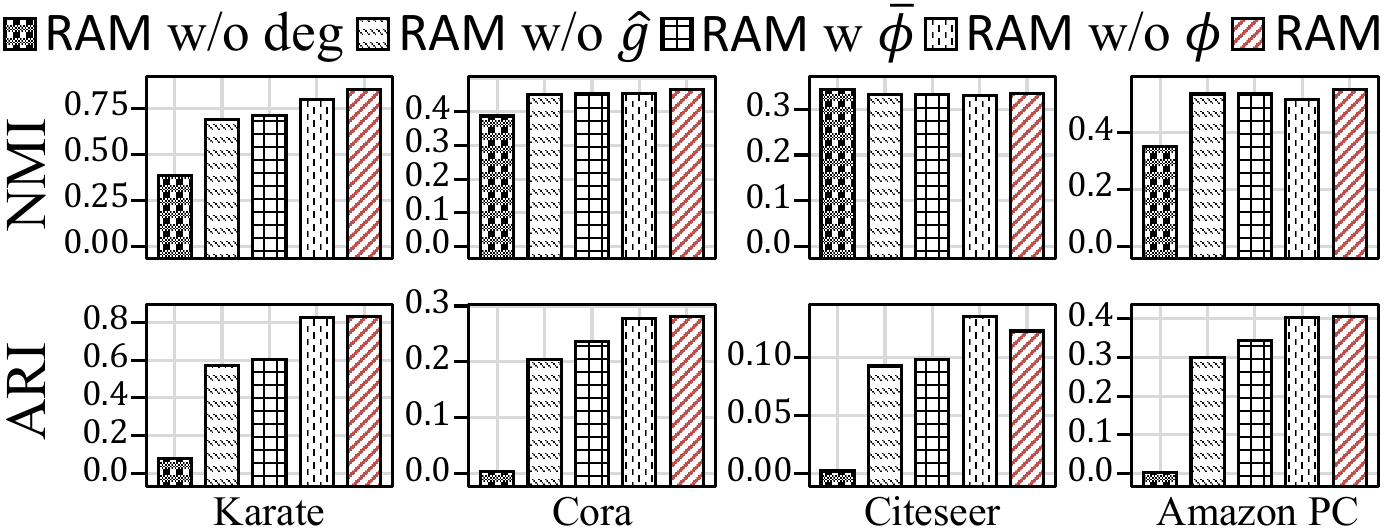}
    \vspace{-0.3cm}
    \caption{EQ5. Distance calibration}
    \label{fig:ablation}
    \vspace{-0.4cm}
\end{figure}

\spara{EQ5. Distance calibration.}
This experiment isolates the contribution of the calibration function $\hat{g}(d_{uv})$ in the {\RAM} null term. We compare full {\RAM} with four ablation variants in Table~\ref{tab:ablation_variants}, all optimised using {\MGA}, and report NMI and ARI in Figure~\ref{fig:ablation}. Across the four datasets, full {\RAM} achieves the highest NMI and ARI in almost all cases. Replacing $\hat{g}(d_{uv})$ with a constant of $1$ ({\RAM} w/o $\hat{g}$) consistently underperforms using the global average affinity ({\RAM} w $\bar{\phi}$), suggesting that both the presence and the quality of the calibration function are important. Removing the degree term ({\RAM} w/o deg) causes the sharpest degradation, confirming the importance of degree heterogeneity; its slightly higher NMI on Citeseer stems from over-fragmentation (1,818 communities, 90\% of nodes in communities of size $\le 3$; ARI~$=0.002$). In contrast, removing $\phi$ from the observed term ({\RAM} w/o $\phi$) yields scores closer to full {\RAM} than the null-term ablations, indicating that geometry-aware calibration mainly contributes through the expected term. Overall, the results support the need to calibrate the null model using both degree heterogeneity and distance-conditioned similarity.


\begin{table}[h]
\centering
\caption{EQ6. Centroid approximation}
\label{tab:centroid_ablation}
\resizebox{\columnwidth}{!}{
\begin{tabular}{l | cc | cc | cc}
\toprule
 & \multicolumn{2}{c|}{\textbf{NMI} ($\uparrow$)} & \multicolumn{2}{c|}{\textbf{ARI} ($\uparrow$)} & \multicolumn{2}{c}{\textbf{Time (s)} ($\downarrow$)} \\
\cmidrule(lr){2-3} \cmidrule(lr){4-5} \cmidrule(lr){6-7}
\textbf{Dataset} & \textbf{Exact} & \textbf{Centroid} & \textbf{Exact} & \textbf{Centroid} & \textbf{Exact} & \textbf{Centroid} \\
\midrule
Karate & 0.8509 & \textbf{0.8519} & \textbf{0.8340} & 0.8289 & 0.0006 & 0.0006 \\
Cora & \textbf{0.4686} & 0.4673 & 0.2777 & \textbf{0.2808} & 1.9659 & \textbf{0.7452} \\
Citeseer & \textbf{0.3380} & 0.3351 & \textbf{0.1305} & 0.1232 & 2.3083 & \textbf{1.2368} \\
Amazon PC & \textbf{0.5562} & 0.5523 & \textbf{0.4252} & 0.4049 & 48.4037 & \textbf{2.4663} \\
\bottomrule
\end{tabular}
}
\end{table}

\spara{EQ6. Centroid approximation.} 
We compare two variants of {\MGA}: \textit{Exact}, which computes $p(u,C)=\sum_{v\in C}p_{uv}$ over all community members, and \textit{Centroid}, which uses the centroid-based approximation $\tilde{p}(u,C)$. We restrict this comparison to small-scale datasets because exact evaluation is computationally expensive on large graphs. Table~\ref{tab:centroid_ablation} reports NMI, ARI, and running time on the four small-scale datasets. Across all datasets, \textit{Centroid} achieves quality comparable to \textit{Exact}, with average gaps of only $0.002$ in NMI and $0.009$ in ARI. At the same time, it substantially reduces overall running time, achieving an average speedup of $6.28\times$. This behaviour is consistent with the error bound in Section~\ref{sec:single}: when community members are coherent in the representation space, the spread $\sigma_C$ remains small, thereby limiting the approximation error. Overall, the centroid approximation preserves clustering quality while improving efficiency, making it a practical choice for scalable {\MGA} optimisation.

\begin{figure}[h]
    \centering
    \begin{subfigure}[h]{\linewidth}
        \centering
        \includegraphics[width=0.99\linewidth]{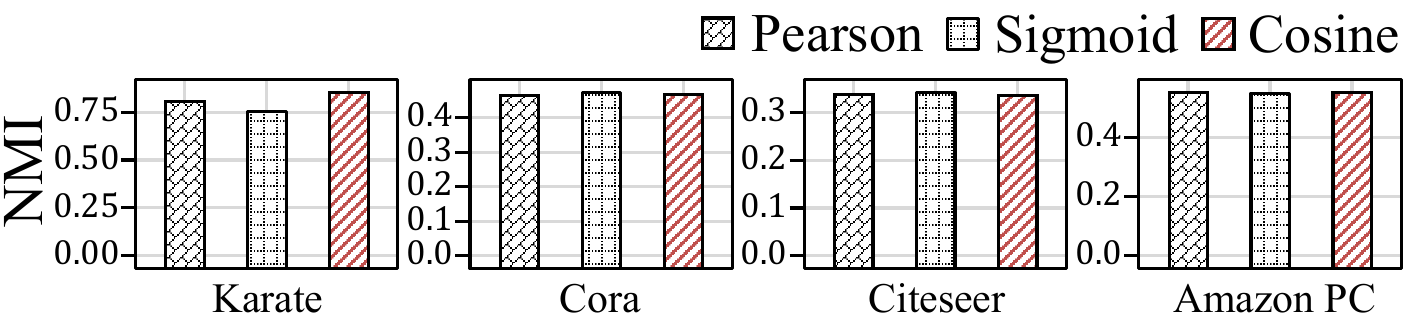}
        \vspace{-0.4cm}
        \subcaption{Variants of similarity}
    \end{subfigure}
    \begin{subfigure}[h]{\linewidth}
        \centering
        \includegraphics[width=0.99\linewidth]{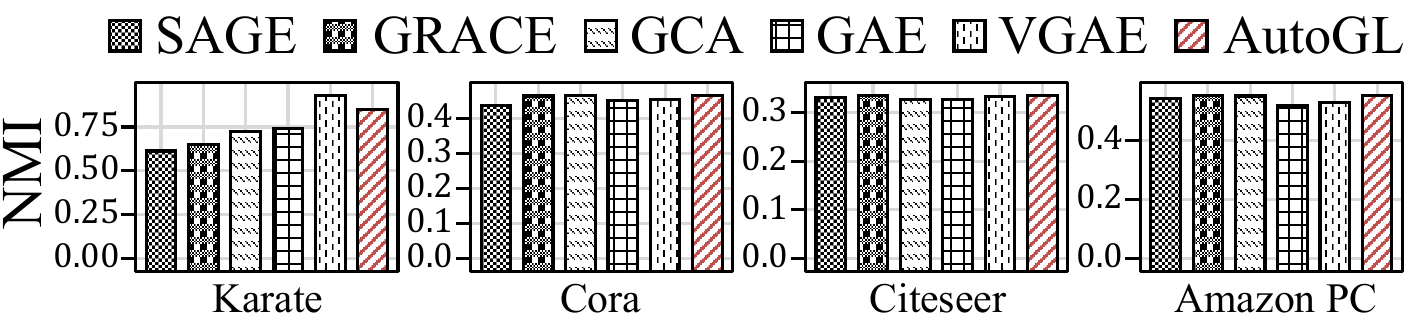}
        \vspace{-0.4cm}
        \subcaption{Variants of embedding}
    \end{subfigure}
    \caption{EQ7. Robustness}
    \label{fig:robustness}
    \vspace{-0.2cm}
\end{figure}

\spara{EQ7. Robustness.} Figure~\ref{fig:robustness} reports the NMI of {\MGA} under different similarity functions and embedding methods on the four small-scale datasets. We first compare three similarity functions: rescaled cosine (default), Pearson correlation, and sigmoid inner product. Since {\RAM} assumes non-negative affinities, we rescale Pearson correlation to $(1+\rho)/2$, analogous to the cosine rescaling. On Cora, Citeseer, and Amazon PC, the three variants differ by less than $0.01$ in NMI, indicating that {\MGA} is reasonably robust to the choice of similarity function. We further test five unsupervised embedding methods: SAGE~\cite{hamilton2017inductive}, GRACE~\cite{zhu2020deep}, GCA~\cite{zhu2021graph}, GAE~\cite{kipf2016variational}, and VGAE~\cite{kipf2016variational}. For each method, we use reported dataset-specific hyperparameters when available, and default settings otherwise. On Cora, Citeseer, and Amazon PC, all embeddings yield NMI within $0.035$ of AutoGL, suggesting robustness to the embedding source when embeddings are sufficiently well-conditioned. Karate shows a wider gap because default hyperparameters are less reliable on such a small graph. AutoGL avoids manual hyperparameter selection and consistently attains competitive NMI, supporting its use as the default embedding pipeline for {\MGA}.


\begin{figure}[h]
    \centering
    \includegraphics[width=0.99\linewidth]{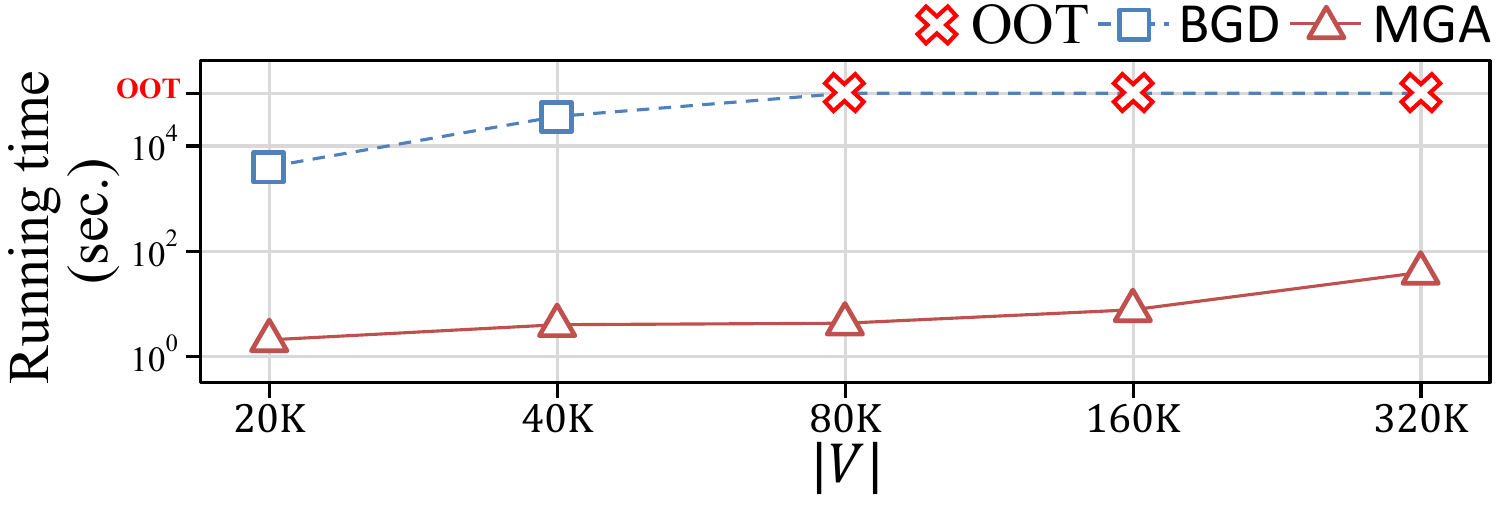}
    \vspace{-0.6cm}
    \caption{EQ8. Scalability}
    \label{fig:scalability}
    \vspace{-0.4cm}
\end{figure}

\spara{EQ8. Scalability.}
To evaluate scalability under controlled graph sizes, we generated synthetic attributed graphs using XMark~\cite{citraro2021x} with its default generation parameters, varying $|V|$ from $20\mathrm{K}$ to $320\mathrm{K}$ by doubling. Figure~\ref{fig:scalability} reports the total running time of {\BGD} and {\MGA}. {\BGD} grows steeply with graph size, taking over $10$ hours already at $|V|=40\mathrm{K}$ and exceeding the $24$-hour time limit beyond that. This behaviour is consistent with the repeated weighted edge betweenness recomputation required by {\BGD}, which incurs a super-linear cost and becomes impractical as the graph grows. In contrast, {\MGA} scales smoothly across the entire range: increasing $|V|$ by $16\times$ ($20\mathrm{K} \to 320\mathrm{K}$) increases the running time by only about $19\times$, showing an approximately linear trend with respect to the generated graph size. This trend is consistent with the multi-level design of {\MGA}, where graph contraction rapidly reduces the size of subsequent optimisation instances.

\section{RELATED WORK}\label{sec:related}
\spara{Modularity and non-standard null models.}
Classical modularity~\cite{newman2004finding} compares observed connectivity with a degree-based null model. Its variants often modify the observed signal or null model using spatial distance~\cite{expert2011uncovering,shakarian2013mining,hannigan2013mining,sarzynska2016null}, hypergraph interactions~\cite{kaminski2019clustering}, or motif-enriched connectivity~\cite{xiao2024higher}. {\RAM} is closest to null-model variants, but assigns a different role to distance. While spatial modularity uses external geographic coordinates to capture spatial dependence, compactness, or dispersion, {\RAM} uses learned representation geometry to calibrate the expected affinity level. The factor $g(d_{uv})$ appears only in the expected term, whereas the observed term remains graph-supported as $A_{uv}\phi(z_u,z_v)$. Thus, {\RAM} is not a direct spatial-modularity analogue, but a representation-aware null model for preserving observed-minus-expected semantics.

\spara{Attribute-aware and representation-based community detection.}
Attribute-aware community detection combines topology with node attributes, either by constructing structural--attribute similarities and partitioning an augmented graph~\cite{zhou2009graph,wu2016community,alinezhad2020community,falih2017anca,berahmand2022new}, or by jointly factorising structural and attribute matrices to infer community assignments~\cite{wang2016semantic,li2017detecting,huang2021community,yang2018enhanced}. Recent representation-based methods further integrate embeddings with community detection by reconstructing the modularity matrix~\cite{yang2016modularity}, adding modularity-aware regularisation to graph autoencoders~\cite{salha2022modularity}, directly optimising modularity through GNNs~\cite{tsitsulin2023graph}, or clustering learned embeddings under modularity-related objectives~\cite{bhowmick2024dgcluster}. These methods use attributes or representations as evidence for assignment, training, or clustering. {\RAM} follows a different paradigm: it treats node representations as input and asks how modularity should be recalibrated when representation-induced affinity is used as the observed signal. Representations therefore do not define a direct clustering space; they calibrate the expected affinity in the null model, so graph-supported affinity is rewarded only when it exceeds what is expected at the corresponding representation distance.
\section{CONCLUSION}
We introduced {\RAM}, an objective for graphs with node representations that preserves observed-minus-expected semantics of classical modularity. The key idea is to calibrate the expected term using both degree heterogeneity and representation distance, evaluating representation-induced affinity against a distance-aware null baseline. We showed that {\RAM} recovers classical modularity as a special case, inherits its hardness, and remains stable under bounded calibration error. Experiments and ablations confirm that distance calibration improves community detection quality and that both calibration components are essential, establishing {\RAM} as an objective-level framework for representation-aware community detection beyond direct embedding-space clustering.


\section*{Acknowledgment}
This work was supported by the Institute of Information \& Communications Technology Planning \& Evaluation(IITP) grant funded by the Korea government(MSIT) (No.RS-2020-II201336, Artificial Intelligence Graduate School Program(UNIST)) and (No. RS-2026-25527734, [Hyper-scale Industrial AI Research Support (R\&D) Program], Development of industry-specified multimodal hyperscale foundation model based technologies for manufacturing, predictive maintenance, and process optimization), and the National Research Foundation of Korea(NRF) grant funded by the Korea government(MSIT) (RS-2025-00523578, RS-2026-25468971)

\spara{AI-generated content acknowledgement.}
Generative AI tools were employed solely for writing assistance and code optimization. Specifically, they were used to improve the grammatical flow of the manuscript and to enhance the readability and efficiency of the implementation code. All core concepts, methodologies, and experimental results were developed independently by the authors, who maintain full responsibility for the content and integrity of this work.

\bibliographystyle{IEEEtran}
\bibliography{sample-base}

\end{document}